\documentclass[conference]{IEEEtran}
\usepackage{cite}
\usepackage{amsmath,amssymb,amsfonts}
\usepackage{algorithmic}
\usepackage{graphicx}
\usepackage{textcomp}
\usepackage{xcolor}
\usepackage{amsthm}
\usepackage{bm}

\usepackage{subcaption}
\usepackage{dblfloatfix}
\usepackage{tikz}
\usepackage{pgfplots}
\pgfplotsset{compat=1.16}

\tikzstyle{projection}=[fill=black, draw=none, shape=circle, scale=0.125]
\tikzstyle{red dot}=[fill=red, draw=black, shape=circle, scale=0.5]
\tikzstyle{blue dot}=[fill=blue, draw=black, shape=rectangle, scale=0.5]
\tikzstyle{red border}=[fill=none, draw=red, shape=circle, scale=0.5]
\tikzstyle{blue border}=[fill=none, draw=blue, shape=rectangle, scale=0.5]
\tikzstyle{green dot}=[fill=green!70!black, draw=black, shape=circle, scale=0.5]
\tikzstyle{black dot}=[fill=black, draw=black, shape=rectangle, scale=0.5]

\tikzstyle{altitude}=[-, densely dashed]
\tikzstyle{red line}=[-, red, line width=1]
\tikzstyle{blue line}=[-, blue, line width=1]
\tikzstyle{green line}=[-, green!70!black, line width=1]
\tikzstyle{black line}=[-, densely dashed, line width=1]
\tikzstyle{segment}=[-, densely dotted, green!35!black, line width=1.25]

\tikzstyle{plane}=[blue!30]
\tikzstyle{polygon}=[blue!20]
\tikzstyle{region}=[lime!20]

\tikzstyle{none}=[]
\pgfdeclarelayer{background}
\pgfdeclarelayer{nodelayer}
\pgfdeclarelayer{edgelayer}
\pgfsetlayers{background,edgelayer,nodelayer}

\newcounter{sectioncnt}[section]

\newtheorem{theorem}{Theorem}[sectioncnt]

\newtheorem{lemma}[theorem]{Lemma}

\def\BibTeX{{\rm B\kern-.05em{\sc i\kern-.025em b}\kern-.08em
    T\kern-.1667em\lower.7ex\hbox{E}\kern-.125emX}}

\newcommand{\rf}{\mathbb{R}}
\newcommand{\wset}{\mathcal{W}}
\newcommand{\arr}{\mathcal{A}}
\newcommand{\upe}{\mathcal{U}}
\newcommand{\lowe}{\mathcal{L}}
\newcommand{\point}{p}
\newcommand{\points}{P}
\newcommand{\plane}{h}

\IEEEoverridecommandlockouts\IEEEpubid{\makebox[\columnwidth]{979-8-3315-2223-0/25/\$31.00~\copyright~2025 IEEE \hfill} \hspace{\columnsep}\makebox[\columnwidth]{ }}

\begin{document}

\title{Robust Bichromatic Classification in 3D Using Planes and Slices}

\author{
\IEEEauthorblockN{Grittin Nuntasombat}
\IEEEauthorblockA{\textit{Department of Computer Engineering} \\
\textit{Kasetsart University}\\
Bangkok, Thailand \\
grittin.n@ku.th}
\and
\IEEEauthorblockN{Nattawut Phetmak}
\IEEEauthorblockA{\textit{Department of Computer Engineering} \\
\textit{Kasetsart University}\\
Bangkok, Thailand \\
nattawut.p@ku.th}
\and
\IEEEauthorblockN{Jittat Fakcharoenphol}
\IEEEauthorblockA{\textit{Department of Computer Engineering} \\
\textit{Kasetsart University}\\
Bangkok, Thailand \\
jtf@ku.ac.th}
}

\maketitle

\begin{abstract}
Given two sets of points in 3-dimensional space $R$ and $B$, 
we want to separate these two sets of points 
using a classifier based on linear constraints, while ensuring robustness against outliers.
The problem was studied in $\rf^2$ by Glazenburg {\em et al.}  
We follow their approach and present various algorithms for many types of classifiers 
under various definitions of outliers.  
Our algorithms rely mainly on the duality of points and planes in $\rf^3$. 
\end{abstract}



\section{Introduction}

Data classification is a fundamental task in data analysis and machine learning.  
When data points are in space, separating two classes of data points (later referred to as {\em red} points and {\em blue} points) can often be viewed as a geometric problem of finding a geometrical object to divide the space into pieces containing different classes of points.  In many cases, the data points cannot be cleanly separated; thus, some measure of separation ``goodness'' has to be employed, either in forms of smooth loss function or the number of points treated as outliers.  

In this work, we consider the problem in low dimension, i.e., when data points are in $\rf^3$, and deal with classification boundaries based on two linear separation.  The quality of the classification is based on the number of outliers.  Following the formulation in~Glazenburg {\em et al.}~\cite{Glazenburg24robust} for $\rf^2$, the problem can be formally stated as follows.

Let $B$ and $R$ be two sets of at most $n$ points in $\rf^3$.   
We refer to points in $B$ as {\em blue points} and points in $R$ as {\em red points}.
Our goal is to find a region, $\wset_B$, bounded by at most two planes, that contains {\em all} blue points while excluding {\em all} red points.  
The points that do not satisfy these conditions are referred to as outliers, and we want to minimize the number of them.
There are two types of outliers.
\begin{itemize}
    \item A {\em red outlier} is a red point contained in the interior of $\wset_B$.
    \item A {\em blue outlier} is a blue point {\em not} contained in $\wset_B$, i.e., a blue point that is in the interior of $\rf^3\setminus\wset_B$. 
\end{itemize}
Given $\wset_B$, we let $k_R$ be the number of red outliers and $k_B$ be the number of blue outliers; we also let $k=k_R+k_B$ be the total number of outliers.  For each type of the region $\wset_B$, we are interested in three problems: (1) minimizing the number of red outliers $k_R$ (while keeping $k_B=0$), (2) minimizing the number of blue outliers $k_B$ (while keeping $k_R=0$), and (3) minimizing $k$ (allowing both types of outliers).

Following~\cite{Glazenburg24robust}, the types of region $\wset_B$ that we consider are
\begin{itemize}
    \item a halfplane,
    \item a {\em slice} bounded by two parallel planes $\plane_1$ and $\plane_2$,
    \item a {\em wedge}, i.e., one of the four subspaces induced by a pair of intersecting planes $\plane_1$ and $\plane_2$, and
    \item a {\em diwedge}, i.e., two opposing subspaces induced by a pair of intersecting planes $\plane_1$ and $\plane_2$.
\end{itemize}
Figure~\ref{fig:region-types} illustrates these four types.  When the type of region is clear, we would refer to a region as $\wset_B(h)$ or $\wset_B(h_1,h_2)$ to emphasize that the region is bounded by $h$ (for halfplane regions) or $h_1$ and $h_2$ for other types of regions.

\begin{figure}[!b]
  \centering
  \begin{subfigure}[t]{0.24\textwidth}
    \centering
    \begin{tikzpicture}[scale=0.4]
    \coordinate (1) at (0, 2) {};
    \coordinate (2) at (4, 0) {};
    \coordinate (3) at (6, 5) {};
    \coordinate (4) at (10, 3) {};
    \begin{pgfonlayer}{background}
        \fill [style=polygon] (1) -- (2) -- (4) -- (10,-1) -- (0,-1) -- cycle;
        \fill [style=plane] (1) -- (2) -- (4) -- (3) -- cycle;
    \end{pgfonlayer}
    \begin{pgfonlayer}{nodelayer}
        \node [style=red dot] (0) at (1.25, 3.75) {};
        \node [style=blue border] (5) at (5.25, 5.25) {};
        \node [style=red dot] (6) at (7.5, 4.75) {};
        \node [style=red dot] (7) at (3.75, 5.75) {};
        \node [style=blue dot] (8) at (8.25, 1.75) {};
        \node [style=blue dot] (9) at (2.5, 1.25) {};
        \node [style=blue dot] (10) at (6.25, 0.5) {};
        \node [style=projection] (11) at (1.25, 2.25) {};
        \node [style=projection] (12) at (3.75, 3) {};
        \node [style=projection] (13) at (5.25, 3) {};
        \node [style=projection] (14) at (7.5, 3.75) {};
        \node [style=projection] (15) at (8.25, 3.25) {};
        \node [style=projection] (16) at (6.25, 2.75) {};
        \node [style=projection] (17) at (2.5, 2.25) {};
        \node [style=projection] (18) at (4.25, 0.75) {};
        \node [style=red dot] (19) at (4.25, 1.5) {};
    \end{pgfonlayer}
    \begin{pgfonlayer}{edgelayer}
        \draw (1) -- (2) -- (4) -- (3) -- cycle;
        \draw [style=altitude] (0) to (11);
        \draw [style=altitude] (17) to (9);
        \draw [style=altitude] (7) to (12);
        \draw [style=altitude] (13) to (5);
        \draw [style=altitude] (10) to (16);
        \draw [style=altitude] (14) to (6);
        \draw [style=altitude] (15) to (8);
        \draw [style=altitude] (18) to (19);
    \end{pgfonlayer}
\end{tikzpicture}
    \caption{halfplane}
    \label{fig:halfplane}
  \end{subfigure}
  \hfill
  \begin{subfigure}[t]{0.24\textwidth}
    \centering
    \begin{tikzpicture}[scale=0.4]
    \node at (0,-0.5) {}; 
    \coordinate (1) at (0, 2) {};
    \coordinate (2) at (4, 0) {};
    \coordinate (3) at (6, 5) {};
    \coordinate (4) at (10, 3) {};
    \coordinate (20) at (0, 5) {};
    \coordinate (21) at (4, 3) {};
    \coordinate (22) at (10, 6) {};
    \coordinate (23) at (6, 8) {};
    \begin{pgfonlayer}{background}
        \fill [style=polygon] (1) -- (2) -- (4) -- (22) -- (23) -- (20) -- cycle;
        \fill [style=plane] (1) -- (2) -- (4) -- (3) -- cycle;
        \fill [style=plane] (20) -- (21) -- (22) -- (23) -- cycle;
    \end{pgfonlayer}
    \begin{pgfonlayer}{nodelayer}
        \node [style=red dot] (0) at (1, 1) {};
        \node [style=blue border] (5) at (1.25, 7.25) {};
        \node [style=red dot] (6) at (7.75, 2.75) {};
        \node [style=red dot] (7) at (3, 7.25) {};
        \node [style=blue border] (8) at (8.75, 1.5) {};
        \node [style=blue dot] (9) at (2.25, 3.5) {};
        \node [style=blue dot] (10) at (5.5, 2) {};
        \node [style=projection] (11) at (1, 2) {};
        \node [style=projection] (12) at (3, 6) {};
        \node [style=projection] (13) at (1.25, 5.25) {};
        \node [style=projection] (14) at (7.75, 3.5) {};
        \node [style=projection] (15) at (8.75, 2.75) {};
        \node [style=projection] (16) at (5.5, 1.25) {};
        \node [style=projection] (17) at (2.25, 2.25) {};
        \node [style=none] (18) at (6.25, 1.5) {};
        \node [style=red dot] (19) at (6.25, 0.5) {};
        \node [style=red dot] (24) at (8.25, 7.25) {};
        \node [style=red dot] (25) at (4.25, 4.75) {};
        \node [style=projection] (26) at (4.25, 3.75) {};
        \node [style=projection] (27) at (8.25, 6.5) {};
        \node [style=projection] (28) at (2.25, 5.25) {};
        \node [style=projection] (29) at (5.5, 4.25) {};
        \node [style=projection] (30) at (6.75, 3.5) {};
        \node [style=projection] (31) at (6.75, 6.5) {};
        \node [style=blue dot] (32) at (6.75, 5.5) {};
    \end{pgfonlayer}
    \begin{pgfonlayer}{edgelayer}
        \draw (1) -- (2) -- (4) -- (3) -- cycle;
        \draw (20) -- (21) -- (22) -- (23) -- cycle;
        \draw [style=altitude] (0) to (11);
        \draw [style=altitude] (7) to (12);
        \draw [style=altitude] (13) to (5);
        \draw [style=altitude] (14) to (6);
        \draw [style=altitude] (15) to (8);
        \draw [style=altitude] (18) to (19);
        \draw [style=altitude] (27) to (24);
        \draw [style=altitude] (26) to (25);
        \draw [style=altitude] (17) to (28);
        \draw [style=altitude] (16) to (29);
        \draw [style=altitude] (30) to (31);
    \end{pgfonlayer}
\end{tikzpicture}
    \caption{slice}
    \label{fig:slice}
  \end{subfigure}
  \hfill
  \begin{subfigure}[t]{0.24\textwidth}
    \centering
    \begin{tikzpicture}[scale=0.4]
    \node at (0,-0.5) {}; 
    \coordinate (1) at (0, 2) {};
    \coordinate (2) at (4, 0) {};
    \coordinate (3) at (6, 5) {};
    \coordinate (4) at (10, 3) {};
    \coordinate (20) at (8, 0) {};
    \coordinate (21) at (6, 6) {};
    \coordinate (22) at (4, 2) {};
    \coordinate (23) at (2, 8) {};
    \coordinate (24) at (7.425, 1.725) {};
    \coordinate (25) at (3.425, 3.725) {};
    \begin{pgfonlayer}{background}
        \fill [style=polygon] (24) -- (4) -- (10,8) -- (23) -- (25) -- cycle;
        \fill [style=plane] (1) -- (2) -- (4) -- (3) -- cycle;
        \fill [style=plane] (20) -- (21) -- (23) -- (22) -- cycle;
    \end{pgfonlayer}
    \begin{pgfonlayer}{nodelayer}
        \node [style=red dot] (0) at (1, 6.75) {};
        \node [style=blue border] (5) at (3, 3.25) {};
        \node [style=red dot] (6) at (8.75, 1.25) {};
        \node [style=red dot] (7) at (2.25, 3.75) {};
        \node [style=blue dot] (8) at (8.25, 6.5) {};
        \node [style=blue dot] (9) at (4.75, 5.5) {};
        \node [style=blue dot] (10) at (7.5, 3.5) {};
        \node [style=projection] (11) at (1, 2) {};
        \node [style=projection] (12) at (2.25, 2.75) {};
        \node [style=projection] (13) at (3, 1.75) {};
        \node [style=projection] (14) at (8.75, 2.75) {};
        \node [style=projection] (15) at (8.25, 3.25) {};
        \node [style=projection] (16) at (7.5, 2.75) {};
        \node [style=projection] (17) at (4.75, 3.75) {};
        \node [style=projection] (18) at (4.75, 0.75) {};
        \node [style=red dot] (19) at (4.75, 2.75) {};
        \node [style=blue border] (26) at (2, 0.5) {};
        \node [style=projection] (27) at (2, 1.5) {};
        \node [style=red dot] (28) at (5.5, 4) {};
        \node [style=projection] (29) at (5.5, 4.5) {};
        \node [style=red dot] (30) at (3.75, 0.5) {};
        \node [style=projection] (31) at (3.75, 1.5) {};
        \node [style=blue border] (32) at (6.25, 3) {};
        \node [style=projection] (33) at (6.25, 4) {};
    \end{pgfonlayer}
    \begin{pgfonlayer}{edgelayer}
        \draw (1) -- (2) -- (4) -- (3) -- cycle;
        \draw (20) -- (21) -- (23) -- (22) -- cycle;
        \draw (24) -- (25);
        \draw [style=altitude] (0) to (11);
        \draw [style=altitude] (17) to (9);
        \draw [style=altitude] (7) to (12);
        \draw [style=altitude] (13) to (5);
        \draw [style=altitude] (10) to (16);
        \draw [style=altitude] (14) to (6);
        \draw [style=altitude] (15) to (8);
        \draw [style=altitude] (18) to (19);
        \draw [style=altitude] (27) to (26);
        \draw [style=altitude] (29) to (28);
        \draw [style=altitude] (30) to (31);
        \draw [style=altitude] (32) to (33);
    \end{pgfonlayer}
\end{tikzpicture}
    \caption{wedge}
    \label{fig:wedge}
  \end{subfigure}
  \hfill
  \begin{subfigure}[t]{0.24\textwidth}
    \centering
    \begin{tikzpicture}[scale=0.4]
    \node at (0,-0.5) {}; 
    \coordinate (1) at (0, 2);
    \coordinate (2) at (4, 0);
    \coordinate (3) at (6, 5);
    \coordinate (4) at (10, 3);
    \coordinate (20) at (8, 0);
    \coordinate (21) at (6, 6);
    \coordinate (22) at (4, 2);
    \coordinate (23) at (2, 8);
    \coordinate (24) at (7.425, 1.725);
    \coordinate (25) at (3.425, 3.725);
    \begin{pgfonlayer}{background}
        \fill [style=polygon] (24) -- (4) -- (10,8) -- (23) -- (25) -- cycle;
        \fill [style=polygon] (24) -- (20) -- (8,-0.5) -- (0,-0.5) -- (1) -- (25) -- cycle;
        \fill [style=plane] (1) -- (2) -- (4) -- (3) -- cycle;
        \fill [style=plane] (20) -- (21) -- (23) -- (22) -- cycle;
    \end{pgfonlayer}
    \begin{pgfonlayer}{nodelayer}
        \node [style=red dot] (0) at (1, 5.5) {};
        \node [style=blue border] (5) at (3.25, 2.5) {};
        \node [style=red dot] (6) at (9, 1.75) {};
        \node [style=red dot] (7) at (2.5, 3.5) {};
        \node [style=blue dot] (8) at (8.25, 7.25) {};
        \node [style=blue dot] (9) at (4.75, 6) {};
        \node [style=blue dot] (10) at (6.5, 3.75) {};
        \node [style=projection] (11) at (1, 2) {};
        \node [style=projection] (12) at (2.5, 2.25) {};
        \node [style=projection] (13) at (3.25, 1.25) {};
        \node [style=projection] (14) at (9, 3) {};
        \node [style=projection] (15) at (8.25, 2.75) {};
        \node [style=projection] (16) at (6.5, 3) {};
        \node [style=projection] (17) at (4.75, 3.75) {};
        \node [style=projection] (18) at (4.25, 0.75) {};
        \node [style=red dot] (19) at (4.25, 3) {};
        \node [style=blue dot] (26) at (1.75, 0.25) {};
        \node [style=projection] (27) at (1.75, 1.75) {};
        \node [style=red dot] (28) at (5.75, 3.75) {};
        \node [style=projection] (29) at (5.75, 4.5) {};
        \node [style=blue dot] (30) at (5.25, 1) {};
        \node [style=projection] (31) at (5.25, 2) {};
        \node [style=blue border] (32) at (7.5, 2.5) {};
        \node [style=projection] (33) at (7.5, 3.75) {};
    \end{pgfonlayer}
    \begin{pgfonlayer}{edgelayer}
        \draw (1) -- (2) -- (4) -- (3) -- cycle;
        \draw (20) -- (21) -- (23) -- (22) -- cycle;
        \draw (24) -- (25);
        \draw [style=altitude] (0) to (11);
        \draw [style=altitude] (17) to (9);
        \draw [style=altitude] (7) to (12);
        \draw [style=altitude] (13) to (5);
        \draw [style=altitude] (10) to (16);
        \draw [style=altitude] (14) to (6);
        \draw [style=altitude] (15) to (8);
        \draw [style=altitude] (18) to (19);
        \draw [style=altitude] (27) to (26);
        \draw [style=altitude] (29) to (28);
        \draw [style=altitude] (30) to (31);
        \draw [style=altitude] (32) to (33);
    \end{pgfonlayer}
\end{tikzpicture}
    \caption{diwedge}
    \label{fig:diwedge}
  \end{subfigure}
  \caption{Types of regions $\wset_B$ in $\rf^3$ in this paper.  Solid blue square represent a point in $W_B$, while empty blue square represent blue outliers.}
  \label{fig:region-types}
\end{figure}
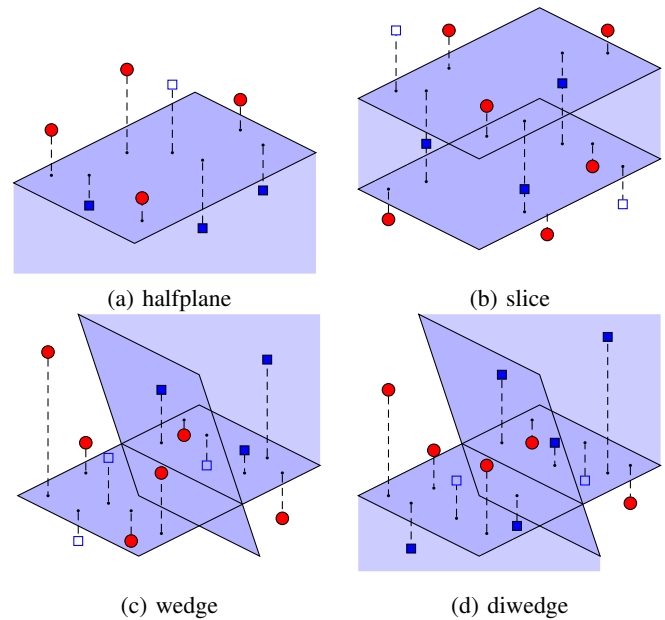

We prove a simple property of the optimal region and provide a polynomial time algorithm for all region types.  
When specializing on halfplanes and slices, we provide faster algorithms.

While most data analysis works with data points in a much higher dimension, recent focuses on explainability and transparency have shifted the toolbox back to simpler model of classification such as decision trees where decision boundaries can be easily visualized.  
While one-attribute threshold-based decision is the simplest form of visualizable boundary, a recent result~\cite{KairgeldinCP24-bivariate} suggested that using two attributes might also be useful.  
This motivates our work for robust classification problems for three attributes, i.e., when data points are in $\rf^3$.

{\bf Related work.} 
In $\rf^2$, these classification problems with outliers have been considered in a unified framework in Glazenburg, van der Horst, Peters, Speckmann, and Staals~\cite{Glazenburg24robust}.  
Finding a line (in $\rf^2$ case) or a plane (in $\rf^3$ case) that separate two sets of points can be solved using a linear-time linear programming algorithm~\cite{Megiddo84-linear-time}.  
The problem of perfectly separating two sets of points in $\rf^2$ using at most two lines has been considered by Hurtado~{\em et al.}~\cite{HurtadoNRS01-sep-wedges,HurtadoMRS04-sep-two-lines}.  For perfect separability of two sets in $\rf^3$, Hurtado~{\em et al.}~\cite{HurtadoSS03-sep3d} also presented algorithms for strips, wedges, and diwedge under various criteria.
When dealing with at most $k$ outliers, finding separating hyperplanes in $\rf^2$ and in $\rf^3$ can be solved using linear programming with violations; Chan~\cite{Chan05-lp-violations} presented algorithms running in time $O((n+k^2)\log n)$ and $O(n\log n + k^{11/4}n^{1/4}\mathrm{polylog}\;n)$, respectively.

{\bf Our contribution.} 
In this paper, we focus on problems in $\rf^3$. 
We present a general algorithm that solves all problems in time $O(n^6)$ based on simple characterizations of optimal solutions.  
For halfplane classifiers and slice classifiers, we also present faster algorithms.  Details of our contribution are in Table~\ref{table:overview}.
We follow the approach by Glazenburg {\em et al.}~\cite{Glazenburg24robust} for problems in $\rf^2$ but extend the ideas into $\rf^3$.  Our algorithms rely heavily on duality as it is typically simpler to search for points instead of planes; Our specialized algorithms also follow the sweep framework. 

\begin{table*}
\centering
\begin{tabular}{ |c|c|c|c||c|c|c|c|c| } 
  \hline
   & \multicolumn{3}{|c||}{2D (all from~\cite{Glazenburg24robust})} & \multicolumn{3}{|c|}{3D} \\ 
      \hline
  Minimize: & Red $k_r$ & Blue $k_b$ & Both $k$ & Red $k_r$ & Blue $k_b$ & Both $k$ \\
  \hline
  \hline
  General & \multicolumn{3}{|c||}{$n^4$} & \multicolumn{3}{|c|}{$\boldsymbol{n}^{\boldsymbol{6}}$; $\boldsymbol{n}^{\boldsymbol{7}}$ for slice} \\ 
  \hline
  Halfplane & \multicolumn{2}{|c|}{$n \log n$} & $(n+k)^ 2\log{n}$ & \multicolumn{2}{|c|}{${\boldsymbol{n}^{\boldsymbol{3}}}$} & $n\log{n} +k^{11/4}n^{1/4} \log^{c}{n}$ \cite{Chan05-lp-violations} \\
  \hline 
  Strip/Slice & $n\log n$ &\multicolumn{2}{|c||}{$n^2 \log n$} & 
  \multicolumn{2}{|c|}{$\boldsymbol{n}^{\boldsymbol{3}} \mathbf{log}\; \boldsymbol{n}$} & $\boldsymbol{n}^{\boldsymbol{5}}$\\
  \hline
  Wedge & $n \log n$ & $n^{5/2}\log n$ & $nk^2\log^3{n}\log{k}$ & \multicolumn{3}{c|}{$\boldsymbol{n}^{\boldsymbol{6}}$  (general algorithm)}\\
  \hline
  Double/di wedges & $n^2$ & $n^2 \log n$ & $n^2k\log^3{n}\log{k}$ & \multicolumn{3}{c|}{$\boldsymbol{n}^{\boldsymbol{6}}$ (general algorithm)} \\
  \hline
\end{tabular}
\caption{Overview of our contribution, shown in bold.  Other results are from the work cited.}
\label{table:overview}
\end{table*}

Section~\ref{sect:prelim} provides background materials.  
We present a general algorithm in Section~\ref{sect:opt-classifiers}.  
Faster algorithms for halfplane classifiers are described in Section~\ref{sect:halfplane-classifiers}.
Finally, we deal with slice classifiers in Section~\ref{sect:slice-classifiers}.


\section{Preliminaries}
\label{sect:prelim}

In this section, we review key duality concepts and state various fundamental results.

\subsection{Duality}

Point-line duality~\cite{BergCKO08-comp-geo} is a standard ingredient in computation geometry, as one can sometimes transform a problem on points and lines (or hyperplanes) into a simpler one. 
We employ the notion in $\rf^3$.

Given point $\point =(a,b,c)$ in $\rf^3$, its dual is a plane $\point^*: z = ax+by-c$.  
On the other hand, for a plane $h: z=a'x+b'y+c'$, its dual is a point $h^* = (a',b',-c')$.  
This dual mapping has an important property: if $\point$ lies {\em above} a plane $\plane$ in the {\em primal} space, 
the plane $\point^*$ lies {\em below} the point $\plane^*$ in the {\em dual} space.
For a set of points $\points$, we are interested in the {\em arrangement} that is constructed by a set of dual plane $\points^*$, $\arr(\points^*)$, i.e. the vertices, edges, faces and cells formed by the planes in $\points^*$.

We denote the {\em upper envelope} of the arrangement $\points^*$, i.e. the surface formed by the highest portions of these planes in $\arr(\points^*)$, by $\upe(\points^*)$, and the {\em lower envelope} by $\lowe(\points^*)$.

We are often interested in the intersection of $\arr(\points^*)$ with another plane $\plane$. The intersection of $\arr(\points^*)$ and $\plane$ is an arrangement of lines. We denote this arrangement of lines as $\arr_\plane(\points^*)$, the upper envelope and the lower envelope of this arrangement as $\upe_\plane(\points^*)$ and $\lowe_\plane(\points^*)$ respectively.

While we work in $\rf^3$, as an example, Fig.~\ref{fig:duality2d} shows points and lines in the primal space and their corresponding lines and points in the dual.

\begin{figure}[!b]
  \begin{subfigure}[t]{0.24\textwidth}
    \centering
    \begin{tikzpicture}[scale=0.15]
    \node at (-8,-12) {};
    \node at (14,12) {};
    \coordinate (B1) at (-4,-8);
    \coordinate (B2) at (2,-10);
    \coordinate (B3) at (6,-3);
    \coordinate (R1) at (12,-2);
    \coordinate (R2) at (-2,4);
    \coordinate (R3) at (6,10);
    \coordinate (GL1) at (-8,-0.8);
    \coordinate (GL2) at (14,-9.6);
    \coordinate (BL1) at (-8,-4.8);
    \coordinate (BL2) at (14,8.4);
    \begin{pgfonlayer}{nodelayer}
        \node [style=blue dot] at (B1) {};
        \node [style=blue dot] at (B2) {};
        \node [style=blue dot] at (B3) {};
        \node [style=red dot] at (R1) {};
        \node [style=red dot] at (R2) {};
        \node [style=red dot] at (R3) {};
    \end{pgfonlayer}
    \begin{pgfonlayer}{edgelayer}
        \draw [style=green line] (GL1) -- (GL2);
        \draw [style=black line] (BL1) -- (BL2);
    \end{pgfonlayer}
    \begin{pgfonlayer}{background}
        \fill [style=region] (GL1) -- (GL2) -- (BL2) -- (BL1) -- cycle;
    \end{pgfonlayer}
\end{tikzpicture}
    \caption{primal}
    \label{fig:primal}
  \end{subfigure}
  \hfill
  \begin{subfigure}[t]{0.24\textwidth}
    \centering
    \begin{tikzpicture}[scale=0.15]
    \node at (-8,-12) {};
    \node at (14,12) {};
    \coordinate (G) at (6,0);
    \coordinate (B) at (-4,4);
    \coordinate (B1) at (-4,-12);
    \coordinate (B2) at (2,-10);
    \coordinate (B3) at (6,-2);
    \coordinate (R1) at (-6,12);
    \coordinate (R2) at (-2,4);
    \coordinate (R3) at (6,10);
    \coordinate (B1L1) at (-8,11.2);
    \coordinate (B1L2) at (14,2.4);
    \coordinate (B2L1) at (-8,8.4);
    \coordinate (B2L2) at (10,12);
    \coordinate (B3L1) at (-8,-1.8);
    \coordinate (B3L2) at (14,11.4);
    \coordinate (R1L1) at (-8,-7.6);
    \coordinate (R1L2) at (25/3,12);
    \coordinate (R2L1) at (-8,2.4);
    \coordinate (R2L2) at (14,-6.8);
    \coordinate (R3L1) at (-10/3,-12);
    \coordinate (R3L2) at (14,-1.6);
    \begin{pgfonlayer}{nodelayer}
        \node [style=green dot] at (G) {};
        \node [style=black dot] at (B) {};
    \end{pgfonlayer}
    \begin{pgfonlayer}{edgelayer}
        \draw [style=blue line] (B1L1) -- (B1L2);
        \draw [style=blue line] (B2L1) -- (B2L2);
        \draw [style=blue line] (B3L1) -- (B3L2);
        \draw [style=red line] (R1L1) -- (R1L2);
        \draw [style=red line] (R2L1) -- (R2L2);
        \draw [style=red line] (R3L1) -- (R3L2);
        \draw [style=segment] (G) -- (B);
    \end{pgfonlayer}
\end{tikzpicture}
    \caption{dual}
    \label{fig:dual}
  \end{subfigure}

  \caption{Duality in $\rf^2$.  The red and blue points and two lines in the primal are mapped to red and blue lines and two points in the dual.}
  \label{fig:duality2d}
\end{figure}
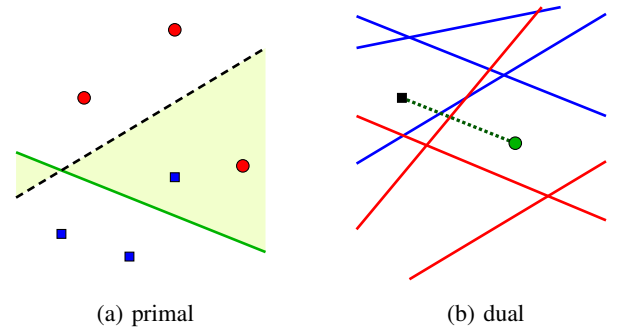

\subsection{Sweep framework}

Our algorithms use the sweep framework.  We only give a quick overview here; for more comprehensive treatment, please see~\cite{BergCKO08-comp-geo}.    

When dealing with many geometrical objects, it is easier to work with them on spaces with smaller dimensions.  
For example, in $\rf^2$, instead of dealing with a number of lines, consider their intersections with a vertical line $x=t$ for some $t$, the lines can be viewed as points.  
An algorithm may ``sweep'' the vertical line $x=t$ through the plane by gradually increasing $t$, while working with these lines as moving points.  
The algorithm designer can thus focus on a one-dimensional problem (on the vertical line) and deal exclusively with how these moving points (from many lines) interact.  
The algorithm generally becomes an event-driven one whose events take place when geometrical objects interact with one another.   

To implement the sweep algorithm efficiently, we often require efficient event queue implementation based on priority queues and other data structures for maintaining the intersections of objects.  Since we deal with objects in $\rf^3$, we sweep a plane $x=t$ through the space.  The geometrical objects that we work with are planes; so their intersections with the sweeping plane become a set of lines, forming an arrangement.


\subsection{Data structures}

To implement our algorithms, we need standard data structures, such as binary trees and priority queues.  
Moreover, as we consider intersections of a set of planes with a sweep plane to obtain an arrangement, we require a data structure suitable for maintaining points and edges with topological information.  
Doubly Connected Edge List (or, DCEL)~\cite{BergCKO08-comp-geo} is a data structure used to represent planar subdivisions, such as polygons or triangulations, that contains vertices, edges, and faces.  
It maintains a set of vertices (possibly with coordinates) and, for each vertex, a counterclockwise list of adjacent edges.  Each edge is represented as two back-and-forth half-edges.  
Each half edge also stores a reference to its adjacent face, while each face keeps a list of half-edges on its boundary in counterclockwise ordering.
The DCEL supports edge traversal, face boundary exploration, and can also be efficiently modified dynamically.



\section{Properties of an optimal classifier}
\label{sect:opt-classifiers}
In this section, we prove properties of optimal $\wset_B$ for each type of regions to motivate our algorithms.  We will work mostly in the dual; the next subsection provides a review.

\subsection{Points and planes in the dual}
\label{sect:dual-problem-formulation}

We describe our settings in the dual.  Since points become planes and planes become points, each point $p\in R\cup B$ becomes a plane $p^*$.  A primal plane $h$ below $p$ becomes a point $h^*$ above the dual plane $p^*$.  

To see that duality can be useful, we refer to Fig.~\ref{fig:duality2d} for an example in $\rf^2$.  Note that the green line is below all blue points and one red point; its dual, the blue dual point lies above all blue lines and one red line.  For the black primal line, since it lies above all red points and one blue point, its dual black point is above all red lines and one blue line.

For concreteness, suppose that we want to find the best halfplane classifier.  We assume that region $\wset_B$ is defined to be the halfspace ``above'' some decision plane.  (To find a classifier defined to be the halfspace ``below'' a plane, one can just flip the points vertically.)   In this case, when searching for plane $h$ defining $\wset_B$, we know that a dual point $h^*$ that lies above a blue dual plane $p^*$ would correctly contain a primal blue point $p$.  Therefore, when we minimize the number of red outliers (while keeping $k_B=0$), we find the dual point $h^*$ that is above {\em all} blue planes while minimizing the number of red planes below $h^*$.  On the other hand, if our goal is to minimize the number of blue outliers (keeping $k_R=0$), the target dual point $h^*$ has to be below {\em all} red planes maximizing the number of blue planes below $h^*$.  Table~\ref{table:dual-problems-halfplane} describes the dual versions of the problems when $\wset_B$ is a halfplane.

\begin{table}
\centering
\begin{tabular}{|p{2cm}|p{5.5cm} |}
    \hline 
    Minimizing & Problems \\
    \hline
    Red outliers $k_R$ & 
    {\bf Primal:}
    Find plane $h$ above all blue points,
    minimizing the number of red points below $h$. \newline
    {\bf Dual:}
    Find point $h^*$ above all blue planes,
    minimizing the number of red planes below $h^*$. \\
    \hline
    Blue outliers $k_B$ & 
    {\bf Primal:}
    Find plane $h$ below all red points,
    maximizing the number of blue points above $h$. \newline
    {\bf Dual:}
    Find point $h^*$ below all red planes,
    maximizing the number of blue planes above $h^*$. \\
    \hline
    Both outliers $k_R+k_B$ & 
    {\bf Primal:}
    Find plane $h$,
    minimizing the sum of the number of red points below $h$
    and the number of blue points above $h$. \newline
    {\bf Dual:}
    Find point $h^*$,
    minimizing the sum of the number of red planes below $h$
    and the number of blue planes above $h$. \\
    \hline
\end{tabular}
\caption{Dual problems for halfplane classifiers}
\label{table:dual-problems-halfplane}
\end{table}

For regions bounded by two planes $\wset_B(h^*_1,h^*_2)$, the same approach works.  In this case, we want to find two dual points $h^*_1$ and $h^*_2$ satisfying particular conditions.  For example, when $\wset_B$ is a slice region, $h^*_1$ and $h^*_2$ are on the same vertical line, and the segment $\overline{h^*_1 h^*_2}$ intersects every dual plane corresponding to points in the region $\wset_B(h^*_1,h^*_2)$.

\subsection{Optimal classifiers}

From the discussion above, in this section, we prove canonical forms of the optimal solutions for various types of regions.  The first lemma deals with slice classification.

\begin{lemma}
\label{lem:3pointsforslice}
For any optimal solution $\wset_B(\plane_1,\plane_2)$ for a slice classification problem, there exists an optimal solution $\wset_B(h_1',h_2')$ such that one and only one of the following is true:
\begin{itemize}
    \item In the primal, one of the plane intersects at least 1 blue point and at least 1 red point. The other plane intersects at least 1 blue point and 1 point of any color, or
    \item In the primal, one of the plane intersects at least 1 blue point, at least 1 red point and 1 point of any color. The other plane intersects at least 1 blue point.
\end{itemize}
\end{lemma}
\begin{proof}
We note that in an optimal solution $\wset_B(h_1,h_2)$ contains at least one blue point.
We move the two planes toward each other until each plane hit a blue point (stop moving if a plane already hits a blue point). 
A plane cannot hit a red point before a blue point, as it would give a better solution with one fewer outlier. 

We can then project every point onto a plane perpendicular to both $h_1$ and $h_2$; on this projection $h_1$ and $h_2$ are now parallel lines.  We can then rotate the two projected lines clockwisely, each using a blue point as a point of rotation, until one of the lines hit a red point; this rotation in the projected plane gives the rotation of $h_1$ and $h_2$ in the original space.  If we hit a blue point before that, use the new blue point as a new point of rotation.  If we have at least one red point, this process surely terminates.

Finally, consider we can rotate the two planes in 3D using the projected lines as the rotation axes until one of them hits another point; this proves the lemma. 
\end{proof}

The following lemma regarding halfplane classifiers can be proved similarly.

\begin{lemma}
    \label{lem:cannonical-halfplane}
    For an optimal solution $\wset_B(h)$ for halfplane classification, there exists an optimal solution $\wset_B(h')$ where $h'$ intersects at least 3 points.
\end{lemma}

The next lemma considers (di)wedge classifiers.

\begin{lemma}
\label{lem:general_algo}
For any optimal solution $\wset_B(\plane_1,\plane_2)$ for a (di)wedge classification problem, there exists an optimal solution, $\wset_B(\plane_1',\plane_2')$, where each of $\plane_1'$ and $\plane_2'$ intersects 3 points with at least 1 red point and at least 1 blue point.
\end{lemma}
\begin{proof}
Let $\plane_1$ and $\plane_2$ be the two planes defining an optimal solution to a (di)wedge classification problem that has at least one red point correctly classified. 
We first fix $\plane_2$ and solve for $\plane'_1$; finding $h'_2$ can be done similarly.  Initially, we let $h'_1=h_1$.
We project every point on the $\rf^3$ space into a plane $V$ that is perpendicular to $\plane'_1$ so that $\plane'_1$ becomes a line $l_1$ in $V$. 
Using the same argument from~\cite{Glazenburg24robust}, we can move $l_1$ down until it hits the convex hull of blue points and rotate $l_1$ around the convex hull until it hit a red point; thus, 
we obtain the intersection with the first 2 points of different colors as claimed.  After that, rotate $\plane'_1$ in $\rf^3$ about the line that is induced by these 2 points until it hit another point. 
If we did not hit another point, then clearly, the entire space only consists of the 2 points that we hit, in which case, the solution is trivial. 
\end{proof}

For all the cases, the fact that $\plane_1$ and $\plane_2$ can be identified by 3 points in the dual space, motivate us to find the solution in the dual space rather than the primal space.  The next lemma considers blue outlier minimization problems.

\begin{lemma}
\label{lem:red_cell}
For the classification problem that minimizes blue outliers, there exists a solution $\wset_B(\plane_1',\plane_2')$ which in the dual, the line segment $\overline{\plane_1'^*,\plane_2'^*}$ is inside a cell in $\arr(R^*)$ and both points $\plane_1'^*$,$\plane_2'^*$ lie on the surface of that cell.
\end{lemma}
\begin{proof}
Let $h_1$ and $h_2$ be two planes that define the optimal solution $\wset_B(\plane_1,\plane_2)$.  
Consider their dual $h^*_1$ and $h^*_2$.  Note that the segment $\overline{h^*_1 h^*_2}$ does not intersect any red planes; i.e., it lies inside some cell in $\arr(R^*)$.
From $\overline{h^*_1 h^*_2}$, we can construct the claimed $\overline{h'^*_1 h'^*_2}$ by extending $\overline{h^*_1 h^*_2}$ on both ends until it reaches other planes.  The extended segment cannot hit a blue plane because it contradicts the optimality of $\wset_B(\plane_1,\plane_2)$; thus, both $h'^*_1$ and $h'^*_2$ are on some pair of red planes as claimed.
\end{proof}

The following lemma deals with red outliers.

\begin{lemma}
\label{blue_surface}
For the classification problem of red outliers $\wset(\plane_1,\plane_2)$, there exists an optimal solution $\wset_B(\plane_1',\plane_2')$ which in the dual space, $\plane_1^*$ lies on the surface of $\upe(B^*)$ and $\plane_2^*$ lies on the surface of $\lowe(B^*)$
\end{lemma}
\begin{proof}
In the primal space, any valid solution to the classification problem of red outliers must contain the convex hull of blue points. In the dual space, this means that $\plane_1^*$ must lie on or above the $\upe(B^*)$ and $\plane_2^*$ must lie on or below the $\lowe(B^*)$. Using the same argument in \ref{lem:red_cell} and shrink the line segment $\overline{\plane_1^*\plane_2^*}$ until $\plane_1^*$ is on the surface of $\upe(B^*)$ and $\plane_2^*$ lies on the surface of $\lowe(B^*)$.
\end{proof}

\subsection{General algorithms}

In this section, we present general algorithms that can solve all problems based on the structures of optimal solutions presented in the previous sections.

For halfplane classification, Lemma~\ref{lem:cannonical-halfplane} simply gives an $O(n^4)$ algorithm by trying all possible 3 points defining the optimal plane; Section~\ref{sect:halfplane-classifiers} gives a faster implementation for red outliers and blue outliers.

For classification problems with two planes for any outliers, we present a general algorithm that can solve using Lemma~\ref{lem:3pointsforslice} and Lemma~\ref{lem:general_algo} that allow us to brute-force through all candidate solutions.  

From both lemmas, we know that an optimal solution must pass through at least one blue point and at least one red point. We claim that both points $\plane_1^*$ and $\plane_2^*$ in the dual space must lie at the intersection of one red plane, one blue plane, and one additional plane of any color.  
To see that, note that there exists an optimal solution where, in the dual space, the points $\plane_1^*$ and $\plane_2^*$ lie on a line formed by the intersection of one red plane and one blue plane. Let $l_1^*$ and $l_2^*$ be the lines that contain the points $\plane_1^*$ and $\plane_2^*$, respectively. We can move $\plane_1^*$ along $l_1^*$ such that it reaches an intersection point of three planes, since this movement does not create any new outliers. Applying the same movement to $\plane_2^*$, we have proved the claim.


There are at most $O(n^3)$ intersection points of three planes; this implies $O(n^6)$ pairs of intersection points. For each pair, we can calculate the number of outliers in $O(n)$ time, which would result in an $O(n^7)$ algorithm.  

However for wedge and (di)wedge classifiers, we can derive a faster algorithm as follows.  First, construct $\arr(B^* \cup R^*)$ in $O(n^3)$ time~\cite{10.5555/2408018}, and select two arbitrary points, $\plane^*_a$ and $\plane^*_b$, to form a line segment $\overline{\plane^*_a\plane^*_b}$. We can then calculate the outliers for this classifier $\wset_B(h_a,h_b)$ in $O(n)$ time.  We can then iterate $\plane^*_a$ through $\arr(B^* \cup R^*)$ in a depth-first search fashion. Note that  when we move from one intersection point to another, only a constant number of outliers can be created; therefore, we can update the information in $O(1)$ time.   We can iterate $\plane_2^*$ through $\arr(B^* \cup R^*)$ similarly to consider every pair of intersection points.  Since there are $O(n^6)$ pairs of points, and each update takes only constant time, we achieve an overall time complexity of $O(n^6)$.

\begin{theorem}
Given two sets of points in $\rf^3$, $R$ and $B$, each containing at most $O(n)$ points, we can construct a classifier $\wset_B$ that is a (di)wedge, minimizing either $k_B$, $k_R$, or $k$ in $O(n^6)$ time.
\end{theorem}


\section{Halfplane Classifier}
\label{sect:halfplane-classifiers}
  
In this section, we present a faster algorithm that finds $\wset_B$ for halfplane classifiers for minimizing red outliers and for minimizing blue outliers.  (See Table~\ref{table:dual-problems-halfplane} for primal and dual problems in this case.)
We only present an algorithm for minimizing red outliers $k_R$; the algorithm for minimizing blue outliers $k_B$ is symmetric.

In the dual, a plane $h$ defining a halfplane classifier corresponds to a point $h^*$. 
From Lemma~\ref{blue_surface}, we know that this point must lie on $\upe(B^*)$. 
Let $r^* \in R^*$ be a red plane in the dual space. 
Consider an intersection of $r^*$ and $\upe(B^*)$ with the $xz$ plane or $yz$ plane.
The result is a plane with red lines and the cross-section of the $\upe(B^*)$.
Each red plane is considered to be a 
\textit{stabbing} plane if inside the intersection plane, the red line that this plane created intersects $\upe(B^*)$ twice,
\textit{splitting} if it intersects $\upe(B^*)$ once, and 
\textit{floating} if it does not intersect $\upe(B^*)$. 

If a plane is splitting, it will create a straight line on $\upe{U}(B^*)$; if it is stabbing, a convex region; 
and if the plane is floating, nothing. 
For a floating plane, we can do nothing to it, since in the primal space this corresponds to a point inside the convex hull of $B$, and we always misclassify this point for a red outlier classification.

To find the optimal plane $h^*$ in the dual, we first construct $\upe(B^*)$ and all its intersections with the red planes.
We can find $\upe(B^*)$ in $O(n^3)$ time by using the fact that $\upe(B^*)$ is created by only the points from the upper hull in the primal space, which can be found in $O(n \log n)$ \cite{convexHull}.  
Since the upper hull has $O(n)$ complexity, then $\upe(B^*)$ can be created in $O(n^3)$ by simply iteratively computing the intersection of corresponding dual planes. 
To deal with red planes, we create $O(n)$ lines and regions induced by the splitting and stabbing planes in $O(n^2)$ time by incrementally intersecting each red plane with $\mathcal{U}(B^*)$.

To find $h^*$, note that $h^*$ belongs to some point on the intersection of $\upe(B^*)$ and the red planes.  We thus traverse all the cells to find it.  We start with an arbitrary vertex in $\upe(B^*)$ and calculate the number of outliers in $O(n)$ time. Then, in a depth-first search fashion, we traverse all the edges $\upe(B^*)$.  Each time we come across a vertex that has a red plane intersecting it, we update the number of outliers based on the slope of the red plane that intersects this edge.  This can be done in constant time for each intersection.
There are $O(n^3)$ vertices and $O(n^3)$ event points in $\mathcal{U}(B^*)$; thus, the algorithm runs in $O(n^3)$ time.

\begin{theorem}
Given two sets of points in $\rf^3$, $R$ and $B$, each containing at most $O(n)$ points, we can construct a classifier $\wset_B$ that is a halfplane and minimizes either $k_B$ or $k_R$ in $O(n^3)$ time.
\end{theorem}


\section{Slice Classifier}
\label{sect:slice-classifiers}

Recall that a region $\wset_B$ defined by a slice classifier is a region between two parallel planes $h_1$ and $h_2$.  

In this section, we show an algorithm for each case of outliers.  The key observation is that $\wset_B$ is a straight line vertical segment in the dual, i.e., its end points $\plane_1^*$ and $\plane_2^*$ share the same $x$ and $y$ coordinates. 

\subsection{Red Outliers}
\label{sect:slice_red}
From Lemma~\ref{blue_surface}, we know that $\plane_1^*$ must lie on $\upe(B^*)$ and $\plane_2^*$ must lie on $\lowe(B^*)$.
Therefore, we use a sweeping plane to find $\overline{\plane_1^*\plane_2^*}$. 
 
Let $x^*,y^*$ and $z^*$ be the three axes of the dual space. 
We set the sweeping plane to be initially $x^*=-\infty$; and sweep through the $x^*$-axis.  The intersection of all planes inside the dual and the sweeping plane create an arrangement of lines $\arr_{x^*}(B^* \cup R^*)$ in 2 dimensions. 
If the solution is on the current sweeping plane, $\plane_1^*\in \upe_{x^*}(B^*)$ and $\plane_2^*\in\lowe_{x^*}(B^*)$; hence, we only need to find the $y^*$ coordinate of segment $\overline{h^*_1h^*_2}$ using the 2-dimensional arrangement $\arr_{x^*}(B^*\cup R^*)$.  Figure~\ref{fig:sweeping-plane-3d} illustrates an example of a sweeping plane $x^*=1$ with 2 blue dual planes and 2 red dual planes.  Their intersections with the sweeping plane are also shown as thick lines.

Figure~\ref{fig:sweeping-plane-3d} demonstrates the duality approach to our problems.

\begin{figure}
\centering
    \includegraphics[width=2in]{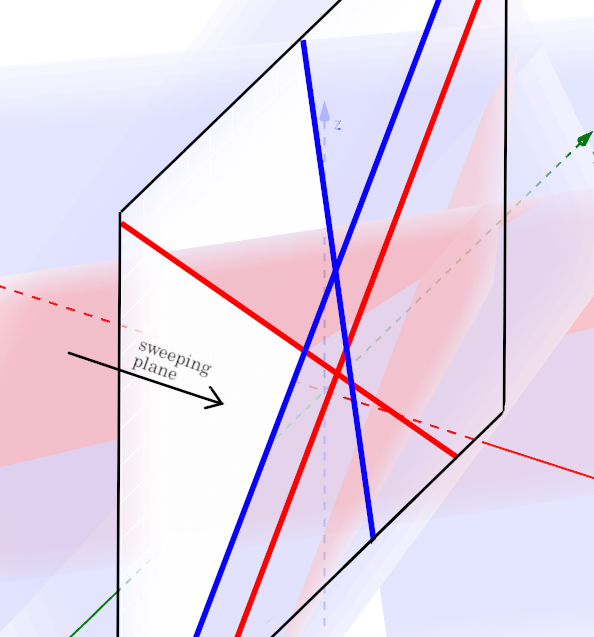}
    \caption{Sweeping plane $x^*=1$ and dual red and blue planes.  Red and blue lines are the intersections of these planes onto the sweeping plane. }
    \label{fig:sweeping-plane-3d}
\end{figure}

As in Section~\ref{sect:halfplane-classifiers},
we classify a red line $r^* \in \arr_{x^*}(B^*\cup R^*)$, corresponding to the intersection of the red plane with the sweeping plane, into three types:
\begin{itemize}
    \item Type \textit{splitting}: 
    $r^*$ intersects both $\upe_{x^*}(B^*)$ and $\lowe_{x^*}(R^*)$. Let $y^*_1$ and $y^*_2$ be the $y^*$ coordinate where that $r^*$ intersects $\upe_{x^*}(B^*)$ and $\lowe_{x^*}(B^*)$ respectively. 
    This means that if our solution lies inside the range $[y^*_1,y^*_2]$, then $r^*$ will always be misclassified.
    \item Type \textit{floating}: 
    $r^*$ does not intersect neither $\mathcal{U}_{x^*}(B^*)$ nor $\mathcal{L}_{x^*}(R^*)$. 
    In the primal space, $r^*$ is inside the convex hull of blue points and always gets misclassified; so we can ignore it.
    \item Type \textit{stabbing}: 
    $r^*$ intersects either one of $\upe_{x^*}(B^*)$ or $\lowe_{x^*}(R^*)$. Note that $r^*$ intersect either $\upe_{x^*}(B^*)$ or $\lowe_{x^*}(R^*)$ only twice;
    let $y^*_1$ and $y^*_2$ be the $y^*$ coordinate of that $r^*$ intersect $\upe_{x^*}(B^*)$ or $\lowe_{x^*}(B^*)$.
    This creates two ranges, $[-\infty,y^*_1]$ and $[y^*_2,\infty]$, in which our solution, when lies inside one of these two ranges, misclassifies $r^*$.
\end{itemize}

As we sweep through the dual space, the problem becomes similar to the line segment intersection, but the key difference is that the ranges are constantly moving. 

Since we are only interested in what interval of $y^*$ yields the minimum intersection of the forbidden ranges, the fact that the ranges are constantly moving only means that there are more event points generated from each range.
There can only be at most $O(n^3)$ vertices in the $\arr(B^* \cup R^*)$ and the event point where two ranges interact (start or stop overlapping each other) is a vertex on the $\mathcal{U}(B^* \cup R^*)$, so the total event points is at most $O(n^3)$.

We set up the event queue by sorting $O(n^3)$ vertices in $O(n^3 \log n)$ time and maintain a priority queue of ranges (sorted by its starting point).
At each event point, we update the priority queue based on the slope of the red plane that intersect the event point. 
This takes at most $O(\log n)$ time. So the total time complexity for this algorithm is $O(n^3 \log n)$. 

\begin{theorem}
Given two sets of points in $\rf^3$, $R$ and $B$, of at most $O(n)$ points each, we can construct a classifier $\wset_B$ that is a slice that minimizes $k_R$ in $O(n^3logn)$ time.
\end{theorem}

\subsection{Blue Outliers}
\label{sect:slice_blue}
From Lemma \ref{lem:red_cell}, we know that our solution is a vertical line $\overline{\plane^*_1\plane^*_2}$ in the dual space such that each endpoint lies on the surface of a cell in $\arr(R^*)$. 
To find the optimal solution, we sweep through $\arr(B^* \cup R^*)$ with a sweeping plane $x^*$ as in the previous section.
During the sweeping process, we would like to find a vertical line $\overline{\plane^*_1\plane^*_2}$ inside a red cell that intersects the maximum number of blue planes.

While in the red outlier case, we deal with moving intersections of lines with the upper and lower hulls, in this case, we deal with changing red cells in $\arr_{x^*}(R^*)$ and moving blue lines in $\arr_{x^*}(B^*)$.
To do so, note that the intersections of dual planes on the sweeping plane $x^*$ is of 2 dimensions; therefore, we use a DCEL data structure to store $\arr_{x^*}(B^* \cup R^*)$ and maintain an active set of red cells,
i.e., those red cells that intersects with the sweeping plane, as faces in the DCEL,
and use the vertices in $\arr(B^* \cup R^*)$ as event points. 
We define four types of vertices, representing various types of event points: 
\begin{itemize}
    \item \textit{RRR} vertices, created by intersections of three red planes.
    \item \textit{RRB} vertices, created by intersections of two red planes and one blue plane.
    \item \textit{RBB} vertices, created by intersections of one red plane and two blue planes.
    \item \textit{BBB} vertices, created by intersections of three blue planes.
\end{itemize}
We update the DCEL according to the type of the event point as follows:
\begin{itemize}
    \item \textit{RRR}: We swap three nodes in the DCEL by selecting an arbitrary node and viewing the other nodes as a supernode. 
    Then, we swap the two nodes, and inside the supernode, we swap the two internal nodes.
    \item \textit{RRB} and \textit{RBB}: We update the number of blue lines inside each face for all faces adjacent to the event point.
    \item \textit{BBB}: In this case, three blue lines swap places, but no lines leave or enter a cell, so we don't need to update the DCEL but we  update the data structure inside each face instead.
\end{itemize}

For each face in the active set, an arbitrary blue line $b_l^*$ that intersect this face at $y^*_1$ and $y^*_2$ will create a range $[y^*_1,y^*_2]$. If $\overline{\plane_1^*\plane_2^*}$ lies inside this range, then the blue plane $b^*$ intersect $\overline{\plane_1^*\plane_2^*}$. 
This means we correctly classified a blue point $b$. We use the same approach as \ref{sect:slice_red} to maintain a correct presentation of ranges but this time we are interested in the range that overlaps the most blue ranges instead of the least.

We show that sweeping through $\mathcal{A}(B^* \cup R^*)$ and updating the DCEL data structure takes, in total, $O(n^3 \log n)$ time, thereby yielding an $O(n^3 \log n)$ algorithm to solve the slice classification for blue outliers.

As an additional data structure, we triangulate all active faces within the sweeping plane and consider an arbitrary face $F$ inside the active set. 
Notice that this face has three vertices. Thus, $F$ must update the number of blue lines inside it at most three times for each blue line that leaves or enters it. 
Each time we update $F$, it takes another $O(\log n)$ time for maintaining a priority queue.  There are $O(n)$ blue lines inside the active set, so updating a face takes $O(n \log n)$ time. 
Since there are $O(n^2)$ faces, the total time for updating the sweeping plane is $O(n^3 \log n)$. Sorting the event points takes $O(n^3 \log n)$. Thus, we achieve the $O(n^3 \log n)$ algorithm as claimed.

\begin{theorem}
Given two sets of points in $\rf^3$, $R$ and $B$, each containing at most $O(n)$ points, we can construct a classifier $\wset_B$, which is a slice that minimizes $k_B$, in $O(n^3 \log n)$ time.
\end{theorem}

\subsection{Both Outliers}


For this case of outliers, recall that Lemma~\ref{lem:3pointsforslice} has two possible cases.  Our algorithm specifically works on each case separately.
The first part of the algorithm deals with the first case of Lemma~\ref{lem:3pointsforslice}, i.e., when one of the plane intersects at least 1 blue point and at least 1 red point and the other plane intersects at least 1 blue point and 1 point of any color.  

The two planes $\plane_1$ and $\plane_2$ are mapped to two points $\plane^*_1$ and $\plane^*_2$ with both of them lies inside some lines, that is an intersection of two planes inside $\arr(B^* \cup R^*)$. 
Since we are looking for a slice classifier, we know that $\plane_1^*$ and $\plane_2^*$ share $x$ and $y$ coordinate. Thus, we project every line in $\arr(B^* \cup R^*)$ to the $xy$ plane.   Under this projection, a slice classifier becomes an intersection point of two lines. 

There are $O(n^2)$ lines inside $\arr(B^* \cup R^*)$; thus, the projection also has $O(n^2)$ lines. 
We can bound the number of intersection points of two lines with $O(n^4)$. An intersection point in the projection corresponds to two point in $\arr(B^* \cup R^*)$ denote by $v_1$ and $v_2$, they share the same $x$ and $y$ coordinate with the intersection point and can be distinguished by the $z$ coordinate that comes from the one of the lines that intersect the intersection point. 
We iterate through all $O(n^4)$ intersection points. For each intersection point, we create a vertical line segment $\overline{v_1v_2}$ in $\arr(B^*\cup R^*)$. We then calculate $k$, the number of outliers, in $O(n)$ time, so we achieve an algorithm with overall running time of $O(n^5)$.

For the other case of Lemma~\ref{lem:3pointsforslice}, one of the endpoints of $\overline{\plane^*_1\plane^*_2}$ is located at one of the vertices inside $\arr(B^* \cup R^*)$ and the other endpoint lies on some blue plane. 
There are at most $O(n^3)$ vertices and $O(n)$ blue planes inside $\arr(B^* \cup R^*)$, the total of$O(n^4)$ possible vertical line segments, determined by drawing a vertical line from the vertex to that blue plane. 
We can calculate the number of outliers $k$ in $O(n)$ time for each vertical line segment, resulting in an overall running time of $O(n^5)$.

\begin{theorem}
Given two sets of points in $\rf^3$, $R$ and $B$, each containing at most $O(n)$ points, we can construct a classifier $\wset_B$, which is a slice that minimizes $k$, in $O(n^5)$ time.
\end{theorem}

\section{Conclusion}

We present algorithms for finding optimal classifiers based on two linear separators in $\rf^3$.  
While general algorithms that use brute-force searching based on the structures of optimal solutions runs in $O(n^7)$ and $O(n^6)$ times, we present algorithms that work in the dual that run much faster in many cases.  
We leave the cases for wedge and (di)wedge classifiers as an important open problem.  

\bibliographystyle{plain} 
\bibliography{refs} 

@InProceedings{Glazenburg24robust,
  author =	{Glazenburg, Erwin and van der Horst, Thijs and Peters, Tom and Speckmann, Bettina and Staals, Frank},
  title =	{{Robust Bichromatic Classification Using Two Lines}},
  booktitle =	{35th International Symposium on Algorithms and Computation (ISAAC 2024)},
  pages =	{33:1--33:14},
  series =	{Leibniz International Proceedings in Informatics (LIPIcs)},
  ISBN =	{978-3-95977-354-6},
  ISSN =	{1868-8969},
  year =	{2024},
  volume =	{322},
  editor =	{Mestre, Juli\'{a}n and Wirth, Anthony},
  publisher =	{Schloss Dagstuhl -- Leibniz-Zentrum f{\"u}r Informatik},
  address =	{Dagstuhl, Germany},
  URL =		{https://drops.dagstuhl.de/entities/document/10.4230/LIPIcs.ISAAC.2024.33},
  URN =		{urn:nbn:de:0030-drops-221605},
  doi =		{10.4230/LIPIcs.ISAAC.2024.33}
}

@book{10.5555/2408018,
author = {Edelsbrunner, Herbert},
title = {Algorithms in Combinatorial Geometry},
year = {2012},
isbn = {3642648738},
publisher = {Springer Publishing Company, Incorporated},
edition = {1st}
}

@article{HurtadoMRS04-sep-two-lines,
title = {Separability by two lines and by nearly straight polygonal chains},
journal = {Discrete Applied Mathematics},
volume = {144},
number = {1},
pages = {110-122},
year = {2004},
note = {Discrete Mathematics and Data Mining},
issn = {0166-218X},
doi = {https://doi.org/10.1016/j.dam.2003.11.014},
url = {https://www.sciencedirect.com/science/article/pii/S0166218X04001891},
author = {Ferran Hurtado and Mercè Mora and Pedro A. Ramos and Carlos Seara}
}

@article{HurtadoNRS01-sep-wedges,
title = {Separating objects in the plane by wedges and strips},
journal = {Discrete Applied Mathematics},
volume = {109},
number = {1},
pages = {109-138},
year = {2001},
note = {14th European Workshop on Computational Geometry},
issn = {0166-218X},
doi = {https://doi.org/10.1016/S0166-218X(00)00230-4},
url = {https://www.sciencedirect.com/science/article/pii/S0166218X00002304},
author = {Ferran Hurtado and Marc Noy and Pedro A. Ramos and Carlos Seara}
}

@InProceedings{HurtadoSS03-sep3d,
author="Hurtado, Ferran
and Seara, Carlos
and Sethia, Saurabh",
editor="Kumar, Vipin
and Gavrilova, Marina L.
and Tan, Chih Jeng Kenneth
and L'Ecuyer, Pierre",
title="Red-Blue Separability Problems in 3D",
booktitle="Computational Science and Its Applications --- ICCSA 2003",
year="2003",
publisher="Springer Berlin Heidelberg",
address="Berlin, Heidelberg",
pages="766--775",
isbn="978-3-540-44842-6"
}

@article{Megiddo84-linear-time,
author = {Megiddo, Nimrod},
title = {Linear Programming in Linear Time When the Dimension Is Fixed},
year = {1984},
issue_date = {Jan. 1984},
publisher = {Association for Computing Machinery},
address = {New York, NY, USA},
volume = {31},
number = {1},
issn = {0004-5411},
url = {https://doi.org/10.1145/2422.322418},
doi = {10.1145/2422.322418},
journal = {J. ACM},
month = jan,
pages = {114–127},
numpages = {14}
}

@article{Chan05-lp-violations,
author = {Chan, Timothy M.},
title = {Low-Dimensional Linear Programming with Violations},
journal = {SIAM Journal on Computing},
volume = {34},
number = {4},
pages = {879-893},
year = {2005},
doi = {10.1137/S0097539703439404}
}

@book{BergCKO08-comp-geo,
author = {Berg, Mark de and Cheong, Otfried and Kreveld, Marc van and Overmars, Mark},
title = {Computational Geometry: Algorithms and Applications},
year = {2008},
isbn = {3540779736},
publisher = {Springer-Verlag TELOS},
address = {Santa Clara, CA, USA},
edition = {3rd ed.}
}

@inproceedings{KairgeldinCP24-bivariate,
author = {Kairgeldin, Rasul and Carreira-Perpi\~{n}\'{a}n, Miguel \'{A}.},
title = {Bivariate Decision Trees: Smaller, Interpretable, More Accurate},
year = {2024},
isbn = {9798400704901},
publisher = {Association for Computing Machinery},
address = {New York, NY, USA},
url = {https://doi.org/10.1145/3637528.3671903},
doi = {10.1145/3637528.3671903},
booktitle = {Proceedings of the 30th ACM SIGKDD Conference on Knowledge Discovery and Data Mining},
pages = {1336–1347},
numpages = {12},
location = {Barcelona, Spain},
series = {KDD '24}
}

@article{convexHull,
author = {Preparata, F.P. and Hong, Se},
year = {1977},
month = {02},
pages = {87-93},
title = {Convex Hull of a Finite Set of Points in Two and Three Dimensions},
volume = {20},
journal = {Communications of the ACM},
doi = {10.1145/359423.359430}
}



\end{document}